\documentclass[a4paper, 11pt]{amsart}
\usepackage{longtable, stackrel, multicol, fancybox, tcolorbox}
\usepackage{bm,amsfonts,amsmath,amssymb,mathrsfs,bbm,amsthm} 
\usepackage{graphicx, color,hyperref,eurosym, eucal,palatino}

\definecolor{darkblue}{rgb}{0.2,0.2,0.6}
\definecolor{darkblue2}{rgb}{0.2,0.2,0.9}
\definecolor{superdarkblue}{rgb}{0.2,0.2,0.3}
\hypersetup{colorlinks=true, linkcolor=darkblue,citecolor=darkblue, 
	urlcolor = superdarkblue}
\definecolor{citegreen}{rgb}{0.2,0.2,0.6}

\usepackage{marginnote}

\usepackage{geometry}

\newcommand\PT{\mathcal{PT}}
\newcommand\R{\mathbb{R}}
\newcommand\bb{\beta}

\newcommand\Th{\Theta}

\newcommand\name[1]{{\small\sc#1}}

\newcommand\df{\dot{f}}

\newcommand\limm{\lim_{m\arr \infty}}

\renewcommand\gg{\gamma}

\newcommand\tm{\times}

\newcommand{\eps}{\varepsilon}

\renewcommand{\aa}{\alpha}

\theoremstyle{definition}

\usepackage[normalem]{ulem}

\makeatletter
\newcommand{\vast}{\bBigg@{3}}
\newcommand{\Vast}{\bBigg@{5}}

\usepackage{scrextend}
\changefontsizes{10.34pt}

\newcommand{\eg}{{\it e.g.}\,}

\newcommand{\cf}{{\it cf.}\,}

\renewcommand\and{\qquad\text{and}\qquad}

\newcommand\sm{\setminus}

\newcommand\dl{\delta}

\newcommand{\comm}[1]{}

\renewcommand\aa{\alpha}
\newcommand\lm{\lambda}
\newcommand\s{\sigma}
\newcommand\p{\partial}

\newcommand\omg{\omega}
\newcommand\Omg{\Omega}

\newcommand\ii{{\mathsf{i}}}

\newcommand\arr{\rightarrow}

\newcommand\dd{{\mathsf{d}}}

\newcounter{counter_a}
\newenvironment{myenum}{\begin{list}{{\rm(\roman{counter_a})}}%
{\usecounter{counter_a}
\setlength{\itemsep}{1.ex}\setlength{\topsep}{0.8ex}
\setlength{\leftmargin}{5ex}\setlength{\labelwidth}{5ex}}}{\end{list}}

\usepackage[latin1]{inputenc}
\usepackage[T1]{fontenc}

\numberwithin{figure}{section}
\numberwithin{equation}{section}
\theoremstyle{plain}
\newtheorem*{thm*}{Theorem}
\newtheorem{thm}{Theorem}
\newtheorem{hyp}{Hypothesis}
\newtheorem{lem}[thm]{Lemma}
\newtheorem{prop}[thm]{Proposition}

\newtheorem{dfn}[thm]{Definition}
\theoremstyle{remark}
\newtheorem{remark}[thm]{Remark}
\theoremstyle{plain}

\newcommand\ov{\overline}

\newcommand\sign{{\rm sign\,}}

\def\ov{\overline}

      \def\sC{{\mathfrak C}}

\def\sS{{\mathfrak S}}

      \def\dC{{\mathbb C}}

   \def\dK{{\mathbb K}}   
   \def\dN{{\mathbb N}}   
      \def\dR{{\mathbb R}}

   \def\dZ{{\mathbb Z}}

\def\sfA{{\mathsf A}}      
      
   \def\sfH{{\mathsf H}}   \def\sfI{{\mathsf I}}
   \def\sfK{{\mathsf K}}

\def\sfS{{\mathsf S}}   \def\sfT{{\mathsf T}}   \def\sfU{{\mathsf U}}
\def\sfV{{\mathsf V}}   \def\sfW{{\mathsf W}}

      \def\cC{{\mathcal C}}
      
   \def\cH{{\mathcal H}}   \def\cI{{\mathcal I}}
\def\cJ{{\mathcal J}}      
      
\def\cP{{\mathcal P}}      
   \def\cT{{\mathcal T}}   
      \def\cX{{\mathcal X}}

   \def\sfk{{\mathsf k}}

   \def\sft{{\mathsf t}}

\newcommand{\Tr}{\mathrm{Tr}\,}

\newcommand{\dom}{\mathrm{dom}\,}

\makeatletter
\def\section{\@startsection{section}{1}\z@{.9\linespacing\@plus\linespacing}%
	{.7\linespacing} {\fontsize{13}{14}\selectfont\bfseries\centering}}

\def\subsection{\@startsection{subsection}{1}\z@{.9\linespacing\@plus\linespacing}%
	{.3\linespacing} {\fontsize{10}{12}\selectfont\bfseries}}	
	
\def\paragraph{\@startsection{paragraph}{4}%
	\z@{0.3em}{-.5em}%
	{$\bullet$ \ \normalfont\itshape}}

\@addtoreset{equation}{section}
\makeatother

\renewcommand{\eg}{{\it e.g.}}

\renewcommand{\cf}{{\it cf.}}

\definecolor{DarkGreen}{rgb}{0,0.5,0.1} 

\newcommand\soutD{\bgroup\markoverwith
{\textcolor{DarkGreen}{\rule[.5ex]{2pt}{1pt}}}\ULon}
\newcommand{\Hm}[1]{\leavevmode{\marginpar{\tiny%
$\hbox to 0mm{\hspace*{-0.5mm}$\leftarrow$\hss}%
\vcenter{\vrule depth 0.1mm height 0.1mm width \the\marginparwidth}%
\hbox to
0mm{\hss$\rightarrow$\hspace*{-0.5mm}}$\\\relax\raggedright #1}}}

\makeatother

\begin{document}
\title[The minimally anisotropic metric operator in...]{ \textsc{The minimally anisotropic metric operator in quasi-Hermitian quantum mechanics} }

\author{David Krej\v{c}i\v{r}\'{i}k}
\address{Department of Mathematics, Faculty of Nuclear Sciences and 
	Physical Engineering, Czech Technical University in Prague, Trojanova 13, 12000 Prague 2,
	Czech Republic}
\email{david.krejcirik@fjfi.cvut.cz}

\author{Vladimir Lotoreichik}
\address{Department of Theoretical Physics, Nuclear Physics Institute, 	Czech Academy of Sciences, 25068 \v Re\v z, Czech Republic}
\email{lotoreichik@ujf.cas.cz}

\author{Miloslav Znojil}
\address{Department of Theoretical Physics, Nuclear Physics Institute, 	Czech Academy of Sciences, 25068 \v Re\v z, Czech Republic}
\email{znojil@ujf.cas.cz}

\begin{abstract}
We propose a unique way how to choose a new inner product in a Hilbert space with respect to which an originally non-self-adjoint 
operator similar to a self-adjoint operator becomes self-adjoint.
Our construction is based on minimising a 'Hilbert-Schmidt distance'
to the original inner product among the entire class of admissible inner products. 
We prove that either the minimiser exists and is unique, 
or it does not exist at all. In the former case we derive 
a system of Euler-Lagrange equations by which the optimal 
inner product is determined. 
A sufficient condition for the existence
of the unique minimally anisotropic metric 
is obtained.
The abstract results are supplied by
examples in which the optimal inner product does not coincide 
with the most popular choice fixed through a charge-like symmetry.
\end{abstract}	

\maketitle
\section{Introduction}
%
Quantum mechanics is a fundamental theory of modern science.
In addition to its enormous success in describing physical phenomena
and important technological impact, it is also mathematically exquisite
through the coherent implementation of functional analysis of operators in Hilbert spaces. It is intrinsically linear and self-adjoint:
physical observables are represented by linear self-adjoint operators 
and the time evolution is governed by unitary groups.

The self-adjointness of quantum theory does not mean 
that one need not deal with an analysis of non-self-adjoint operators 
when it comes to applications. The description of quantum scattering 
by means of a complex effective potential due to 
\name{H.~Feshbach} in 1958~\cite{Feshbach_1958}
is just an early example.
However, the non-self-adjointness arises in such approaches 
as a result of a technical method or a useful approximation 
to attack a concrete physical problem involving
observables still represented by self-adjoint operators.

A conceptually new point of view in this respect was suggested 
by nuclear physicists \name{F.\,G.~Scholtz}, \name{H.\,B.~Geyer},
and \name{F.\,J.\,W.~Hahne} in 1992~\cite{GHS}:
physical observables in quantum mechanics can be represented
by non-self-adjoint operators provided that they are 
\emph{quasi-self-adjoint}. 
They actually use the term \emph{quasi-Hermitian},
which was also previously used by the mathematician 
\name{J.~Dieudonn\'{e}} in 1961~\cite{Dieudonne_1961},
but we have decided to use a more modern mathematically terminology in this paper.
An operator~$\sfH$ in the Hilbert space~$\cH$
equipped with the inner product $\langle\cdot,\cdot\rangle$
is called quasi-self-adjoint if it is densely defined
and there exists a bounded, non-negative self-adjoint operator 
$\Th \colon \cH\arr\cH$ having a bounded inverse such that 
\begin{equation}\label{quasi}
  \sfH^* = \Th \sfH \Th^{-1} \,.
\end{equation}
Here~$\sfH^*$ denotes the adjoint of~$\sfH$ in~$\cH$
with respect to the inner product $\langle\cdot,\cdot\rangle$.

A quasi-self-adjoint~$\sfH$ is self-adjoint with respect to a
modified (but topologically equivalent) inner product
$\langle\cdot,\Th\,\cdot\rangle$ in~$\cH$. For this reason, the
operator $\Th$ is often called \emph{metric}.
Let us also remark that the quasi-self-adjointness of~$\sfH$
is equivalent to the fact that~$\sfH$ is similar to 
a self-adjoint operator in~$\cH$ with respect to 
the original inner product $\langle\cdot,\cdot\rangle$.
Consequently, the spectrum of a quasi-self-adjoint operator~$\sfH$
is purely real and~$\sfH$ generates a unitary time evolution
when considered in the Hilbert space with the modified inner product. 
In summary, by considering the larger class of quasi-self-adjoint operators,
one gets a useful flexibility (overlooked for almost one century)
in the mathematical description of physical observables in quantum mechanics.
We refer to the recent book~\cite{BGSZ} 
for more information on this new concept in quantum mechanics
and many references.

It is well known and easy to check that the choice of
the metric operator~$\Th$ in~\eqref{quasi} is inevitably non-unique. 
In fact, this diversity occurs already in the self-adjoint 
situation $\sfH^*=\sfH$ when one has a one-parametric 
sub-family of metrics $\Theta_\alpha := \alpha \sfI$ with $\alpha > 0$. 
A meaningful way of picking up a `good' metric operator
for a given quasi-self-adjoint operator 
is the subject of active ongoing research.
This issue is addressed already in the pioneering work~\cite{GHS},
where the uniqueness is partially settled by considering 
an irreducible set of quasi-self-adjoint operators.
For just one operator, however, there is no canonical way
how to choose the metric operator and the existing procedures
are typically motivated by simplicity of calculation
(see, \eg, \cite{KBZ,K4,KK5})
or extra physical-like symmetries
(see, \eg, \cite{BBJ,Albeverio-2005-38}). 

The goal of this paper is to present a new promising 
way of selecting the metric, 
which relies on a certain minimality condition.
We focus on quasi-self-adjoint operators 
having purely real simple discrete spectra. 
Moreover, we assume  that the eigenfunctions constitute 
a basis quadratically close to an orthonormal 
one \cite[\S VI.2]{Gohberg-Krein_1969}. 
For a quasi-self-adjoint operator~$\sfH$ in the Hilbert space~$\cH$
satisfying these hypotheses, 
it follows from \cite[\S 2.3]{KSZ} and Lemma~\ref{lem:HS} below
that~$\sfH$ possesses a sub-family of metrics 
of the form $\Th = \sfI + \sfK$, 
where $\sfK$ is a Hilbert-Schmidt operator in $\cH$ of a specific
structure. The metric in this sub-family with the smallest
Hilbert-Schmidt norm of $\sfK$ will be called \emph{minimally
anisotropic}.

Our main result formulated in Theorem~\ref{thm:main} shows that
either there exists a unique minimally anisotropic metric, 
or there is no such metric at all. 
In the former case we derive a system of
Euler-Lagrange equations, by which the minimally anisotropic metric is uniquely determined.
Furthermore, we provide in Proposition~\ref{prop:existence} a condition on the eigenfuctions of $\sfH^*$ which is sufficient
for the existence of the unique minimally	 
anisotropic metric for $\sfH$.

The abstract results are supported by finite- and infinite-dimensional
examples. These examples illuminate the mechanism behind
existence/non-existence of the minimally anisotropic metric. They
show that the minimally anisotropic metric need not coincide with
the metric constructed by means of the charge-symmetry operator.

Instead of minimising the Hilbert-Schmidt norm of~$\sfK$,
it is also possible to consider the analogous minimisation problems 
in other Schatten classes (including the operator norm of 
the Hilbert space~$\cH$).
The distinction of our choice is that the class of Hilbert-Schmidt
operators forms a Hilbert-space structure.

\section{Admissible class of quasi-self-adjoint operators}
In what follows $(\cH, \langle\cdot,\cdot\rangle)$ is a separable Hilbert space of dimension $N\in\dN\cup\{\infty\}$. We adopt the physical convention that the inner product is linear in the second entry.
The norm in $\cH$, induced by the inner product $\langle\cdot,\cdot\rangle$, 
will be denoted by $\|\cdot\|$. 
We assume that the reader is familiar with the basics of the Hilbert space theory. Nevertheless, we provide definitions for some selected important concepts and recall their key properties. For the sake of convenience, we define the set 
\[
	\dK := 
	\begin{cases}
	\{1,2,\dots,N\}, & \quad N < \infty,\\
	\dN,             & \quad N = \infty.
	\end{cases}
\]	

A sequence of vectors $\{\phi_n\}_{n=1}^N$ in $\cH$
is called a \emph{basis} if any $\phi\in\cH$ admits a unique
expansion into the series $\phi = \sum_{n=1}^N c_n\phi_n$
with the complex coefficients $\{c_n\}_{n=1}^N$.
This series is assumed to be norm-convergent if $N = \infty$.

Let $\{\chi_n\}_{n=1}^N$ be an orthonormal basis in $\cH$
and let $\sfA$ be a bounded and boundedly invertible operator in $\cH$.
By~\cite[\S VI.2]{Gohberg-Krein_1969}, 
the set of vectors $\phi_n := \sfA \chi_n$, $n\in\dK$,
is also a basis in $\cH$, called a \emph{Riesz basis}. 
Let $\cX = \{\chi_n\}_{n=1}^N$ be an orthonormal basis in $\cH$
and the family $\Psi = \{\psi_n\}_{n=1}^N$ be a basis such that
$\sum_{n=1}^N \|\psi_n - \chi_n\|^2 < \infty$.
Such a basis is called \emph{quadratically close
to an orthonormal} alias \emph{Bari basis}.
By~\cite[\S VI.2, Thm. 2.3]{Gohberg-Krein_1969}, the
Bari basis $\Psi$ is also a Riesz basis with respect to $\cX$.
For $N < \infty$, any basis is a Riesz as well as a Bari basis.

Now we introduce a class of quasi-self-adjoint operators.
\begin{hyp}\label{hyp:1}
	Let $\sfH$ be a quasi-self-adjoint operator in $\cH$ with 
	purely real simple discrete spectrum 
	and assume that the set 
	of its eigenfunctions $\Psi = \{\psi_n\}_{n=1}^N$ is a basis in $\cH$. 
	Additionally, assume that $\Psi$ is quadratically
	close to an orthonormal basis $\cX = \{\chi_n\}_{n=1}^N$.
\end{hyp}

Recall that an operator is said to have a purely discrete spectrum
if its resolvent is compact. The spectrum is said to be simple
if the geometric and the algebraic multiplicities
of all the eigenvalues are equal to one.
While the condition on basis properties is automatically satisfied
only in finite-dimensional Hilbert spaces,
it also holds for a  class of Schr\"odinger operators on a bounded interval with complex Robin
boundary conditions~\cite{KSZ}. 
On the other hand, there exist quasi-self-adjoint operators
without the Bari property~\cite{K13}.

Clearly, the operator $\sfH^*$ has the same simple discrete spectrum as $\sfH$.
Let $\Phi = \{\phi_n\}_{n=1}^N$ be the set of the eigenfunctions of $\sfH^*$. 
We adopt the normalisation for the families $\Psi$ and $\Phi$ such that:
\begin{myenum}
	\item $\langle\psi_m,\phi_n\rangle = \dl_{nm}$ for $n,m \in\dK$;
	\item $\|\phi_n\| = 1$ for $n\in\dK$.
\end{myenum}	
The normalisation condition~(i) 
implies convenient resolution-of-identity decompositions 
\[
	\sfI = \sum_{n=1}^N \psi_n\langle\phi_n,\cdot\rangle = \sum_{n=1}^N \phi_n\langle\psi_n,\cdot\rangle.
\]
According to~\cite[\S VI.3]{Gohberg-Krein_1969}, 
the family $\Phi$ is also a Bari basis in $\cH$,
being quadratically close to $\cX$.
Moreover, there exists a bounded and boundedly invertible operator $\sfA$ in $\cH$ such that 
$\chi_n = \sfA^*\phi_n = \sfA^{-1}\psi_n$ for all $ n \in\dK$.

In view of the construction in~\cite[Sec.~2.3]{KSZ},
any metric operator $\Th$ for the quasi-self-adjoint operator $\sfH$ as in Hypothesis~\ref{hyp:1} admits the following representation: 
\begin{equation}\label{eq:Theta}
	\Th = \sum_{n=1}^N C_n \langle \phi_n, \cdot\rangle \phi_n,
\end{equation}
where $C_- \le  C_n \le C_+$ for all $n\in \dK$
with some $C_-, C_+ \in ( 0, \infty )$, $C_- \le C_+$. 	
Note that the sum in~\eqref{eq:Theta} should be understood as the
strong limit in the case that $N = \infty$.
Evident ambiguity in the choice of the constants $C_n$
in the representation~\eqref{eq:Theta} reflects
the non-uniqueness of the metric.

\section{The Hilbert-Schmidt and the trace classes}
We assume that the reader is familiar with the concept of the 
compact linear operator in a Hilbert space~\cite[\S 2.6]{Birman-Solomyak}.
For a compact linear operator $\sfT \colon \cH\arr\cH$
we define its module by $|\sfT| := (\sfT^*\sfT)^{1/2}$.
The eigenvalues $s_k(\sfT)$, $k\in\dK$, of $|\sfT|$ ordered non-increasingly
and with multiplicities taken into account are called
the singular values of $\sfT$.

A compact operator $\sfT \colon \cH\arr\cH$ is \emph{Hilbert-Schmidt} 
(respectively, of \emph{trace class}) if, and only if, 
$\sum_{k=1}^N (s_k(\sfT))^2 < \infty$ 
(respectively, $\sum_{k=1}^N s_k(\sfT) < \infty$).
We denote 
by $\sS_2(\cH)$ and by $\sS_1(\cH)$ the families of all Hilbert-Schmidt and of all trace class operators over $\cH$, respectively. In particular, the inclusion $\sS_1(\cH)\subsetneq\sS_2(\cH)$ holds. Note that in a finite-dimensional
Hilbert space ($N < \infty$) any liner operator is Hilbert-Schmidt as well as of trace class.

For any operator $\sfT \in \sS_1(\cH)$ we denote by $\lm_k(\sfT)$, $k\in\dK$, its eigenvalues repeated with the algebraic multiplicities taken into account.
The trace mapping
\[
	\sS_1(\cH)\ni \sfT \mapsto \Tr \sfT := \sum_{k=1}^N \lm_k(\sfT)
\]
is well defined and the sum on the right-hand side converges absolutely.
Let $\Omg\subset\dR^d$, $d \ge 1$, be an open set.
The trace of an integral operator $\sfT \in\sS_1(L^2(\Omg))$ with the integral kernel $\sft\colon \Omg\tm\Omg\arr\dR$ satisfying $\sft\in C(\ov{\Omg\tm\Omg})$ can be computed as follows (\cf~\cite[Ex.~X.1.18]{Kato} and~\cite[Cor.~3.2]{Br91})
\[
	\Tr\sfT = \int_\Omg \sft(x,x) \, \dd x.
\]

The class $\sS_2(\cH)$ of Hilbert-Schmidt operators
over $\cH$ viewed as a linear space
and endowed with the conventional inner product 
$\langle\sfS,\sfK\rangle_2 := \Tr(\sfS^*\sfK)$
turns out to be a Hilbert space; \cf~\cite[\S III.9]{Gohberg-Krein_1969}.
The norm on $\sS_2(\cH)$ induced by the inner product $\langle\cdot,\cdot\rangle_2$
will be denoted by $\|\cdot\|_2$.

\section{The main result}
%
In this section we prove the main result of this paper. To this aim we need an auxiliary
lemma on metric operators that can be represented as the sum of the identity and a Hilbert-Schmidt operator.
\begin{lem}\label{lem:HS}
	Let $N = \infty$ and let $\sfH$ be a quasi-self-adjoint operator as in Hypothesis~\ref{hyp:1}.
	Let $\Th$ be a metric operator for $\sfH$ represented as in~\eqref{eq:Theta}.
	Then $\sfI - \Th$ is Hilbert-Schmidt if, and only if, 
	$C_n = 1 + \aa_n$ with $\aa = \{\aa_n\}_{n=1}^\infty \in \ell^2(\dN)$.
\end{lem}	
\begin{proof}
	Define the following auxiliary operators
	\[
		\sfU := \sfA^*\big(\sfI - \Th\big)\sfA,
		\qquad
		\sfV := \sum_{n=1}^\infty\aa_n \chi_n \langle \chi_n, \cdot\rangle,
		\quad\text{and}\quad
		\sfW := 
		\sum_{n=1}^\infty \chi_n \big\langle \sfA^*(\psi_n - \phi_n),\cdot \big\rangle. 
	\]
	The operator $\sfU$ can be represented as
	\[
	\begin{split}
		\sfU 
		& = 
		\sum_{n=1}^\infty
		\left [ 
			\sfA^*\phi_n\langle \sfA^*\psi_n, \cdot\rangle 
			- 
			\big[1+\aa_n\big]\sfA^*\phi_n \langle \sfA^*\phi_n, \cdot \rangle
		\right ]\\
		& = 
		\sum_{n=1}^\infty
		\left [ 
			\chi_n \langle \chi_n, \cdot \rangle
			+
			\chi_n \langle \sfA^*(\psi_n - \phi_n), \cdot \rangle 
			- 
			\big[ 1 + \aa_n \big] \chi_n \langle \chi_n, \cdot \rangle
		\right ]  
		= \sfW - \sfV.
	\end{split}
	\]
	The Hilbert-Schmidt norm of
	$\sfV$ is given by $\|\sfV\|_2 = \|\aa\|_{\ell^2(\dN)}$.
	For the operator $\sfW$ we can estimate
	the square of the Hilbert-Schmidt norm as
	follows
	\begin{equation*}\label{key}
		\|\sfW\|^2_2 
		=
		\sum_{n=1}^\infty \sum_{m=1}^\infty
		\big|\langle\sfA^*(\psi_n - \phi_n),\chi_m\rangle\big|^2
		= 
		\sum_{n=1}^\infty\big\|\sfA^*(\psi_n - \phi_n)\big\|^2
		\le 
		\|\sfA\|^2 \sum_{n=1}^\infty\|\psi_n - \phi_n\|^2 < \infty.
	\end{equation*}
	Suppose that $\sfI - \Th$ is a Hilbert-Schmidt operator. 
	Then the operator $\sfU$ is Hilbert-Schmidt as well and we have 
	\[
		\|\aa\|_{\ell^2(\dN)}^2 
		= 
		\|\sfV\|^2_2 
		\le 
		2\|\sfU\|_2^2 + 2\|\sfW\|^2_2
		\le
		2\|\sfU\|_2^2 + 2\|\sfA\|^2\sum_{n=1}^\infty \|\psi_n - \phi_n\|^2 
		< 
		\infty.
	\]
	\smallskip
	Second, suppose that $\Th$ is as in~\eqref{eq:Theta}
	with $C_n = 1+\aa_n$ where $\aa \in \ell^2(\dN)$.
	Then
	\[
	\begin{split}
		\|\sfI - \Th\|^2_2 
		&\le 
		\|\sfA^{-1}\|^4\|\sfU\|^2_2
		\le 
		2\|\sfA^{-1}\|^4\left (\|\sfW\|^2_2 + \|\sfV\|^2_2\right)	\\
		& \le 
		2\|\sfA^{-1}\|^4\|\sfA\|^2	\sum_{n=1}^\infty\|\psi_n - \phi_n\|^2 
		+
		2\|\sfA^{-1}\|^4\|\aa\|^2_{\ell^2(\dN)} <\infty. \qedhere		
	\end{split}
	\]
\end{proof}	
Next, we introduce the cone of operators 
\[
	\sC 
	:= 
	\left\{
		\sum_{n=1}^N \big[1+ \aa_n\big]\langle\phi_n,\cdot\rangle\phi_n\colon 
		\aa \in \ell^2(\dK;\cI)
	\right\},
	\qquad\text{where}\quad\cI = [-1,\infty).
\]
For $\Th \in\sC$, we call $\aa = \aa(\Th)\in\ell^2(\dK;\cI)$ its \emph{characteristic vector}.
It is straightforward to see that for any $\omg \in [0,1]$ and $\Th_1,\Th_2 \in \sC$ 
we have $\omg\Th_1 + (1-\omg)\Th_2 \in\sC$. 
The boundary $\p\sC$ and the interior $\sC^\circ$ of $\sC$ are defined by
\[
	\p\sC     := \big\{ \Th\in\sC\colon \min\aa(\Th) = -1\big\}\and
	\sC^\circ := \big\{ \Th\in\sC\colon \min\aa(\Th) > -1\big\}.
\]
In view of Lemma~\ref{lem:HS} any metric for $\sfH$, which can be represented
as $\Th = \sfI + \sfK$ with a Hilbert-Schmidt $\sfK$, belongs to $\sC^\circ$.
Now, we consider the minimisation problem:
\begin{equation}\label{eq:problem}
  \inf_{\Th\in\sC} \|\Th - \sfI\|_2 
  \,.
\end{equation}
%

%
\begin{prop}\label{prop:main}                        
	Under Hypothesis~\ref{hyp:1}, the minimisation problem~\eqref{eq:problem}
	has a unique minimiser $\Th_\star\in\sC$.
\end{prop}
%
\begin{proof}
	We split the proof into two steps. First, we show existence of a minimiser, then we 
	prove its uniqueness.
	
	\noindent {\bf\emph{Step 1 (existence).}}
	Suppose that the infimum in~\eqref{eq:problem} equals $M \ge 0$.
	Let $(\Th_m)_m\subset\sC$ be a minimising sequence	that is
	$\dl_m := \| \Th_m - \sfI\|_2^2 - M^2 \arr 0^+$ as $m\arr \infty$.	
	Employing the parallelogram identity in the Hilbert space $\sS_2(\cH)$, we find
	\[
		\| \Th_k - \Th_l\|_2^2 
		= 
		2\|\Th_k - \sfI\|_2^2 + 2\|\Th_l - \sfI\|_2^2
		-
		4\left\|\tfrac{\Th_k + \Th_l}{2} - \sfI\right\|_2^2
		\le
		2\dl_k + 2\dl_l\arr 0^+,\quad k,l\arr \infty.
	\]
	Hence, $\sfI - \Th_m$ is a Cauchy sequence in $\sS_2(\cH)$. 
	Thus, $\sfI - \Th_m$ converges in the norm $\|\cdot\|_2$ to some Hilbert-Schmidt operator $\sfK_\star$. 
	Moreover, by continuity we infer that $\|\sfK_\star\|_2 = M$.
	
	Next, we  show that $\Th_\star = \sfI + \sfK_\star \in \sC$.
	Indeed, for every $n\in\dK$ we have
	$\Th_\star \psi_n = \limm\Th_m\psi_n \in {\rm span}\,\{\phi_n\}$.
	Thus,	we conclude that
	\[
		\Th_\star = \sum_{n=1}^N \big[1+ \aa_n(\Th_\star)\big]\langle\phi_n,\cdot\rangle\phi_n,
	\]
	and it only remains to show that $\aa(\Th_\star) \in \ell^2(\dK;\cI)$.
	Since $\Th_\star - \sfI$ is Hilbert-Schmidt,
	\[
	\begin{split}
		\|\aa(\Th_\star)\|_{\ell^2(\dK)}^2
		& 
		= \sum_{n=1}^N |\langle\psi_n, ( \Th_\star - \sfI )\psi_n\rangle|^2
		= \sum_{n=1}^N |\langle\chi_n, \sfA^* ( \Th_\star - \sfI ) \sfA\chi_n\rangle|^2\\
		& 
		\le \big\| \sfA^* (\Th_\star - \sfI) \sfA \big\|^2_2
		\le \|\sfA\|^4 \| \Th_\star - \sfI \|^2_2  < \infty.
	\end{split}
	\] 
	Finally, we check that $\aa_n(\Th_\star) \in \cI$. Indeed, we have
	\[
		1 + \aa_n(\Th_\star) 
		= \langle \phi_n, \Th_\star \psi_n \rangle 
		= \limm \langle \phi_n, \Th_m \psi_n \rangle 
		= \limm (1+\aa_n(\Th_m)) \ge 0.
	\]
	{\bf \emph{Step 2 (uniqueness).}}
	Suppose that for some $\Th_1,\Th_2\in\sC$
	we have $\| \Th_1 - \sfI\|_2 = \| \Th_2 - \sfI\|_2 = M$.
	For the arithmetic mean $\Th_3 := \frac12(\Th_1 + \Th_2)\in \sC$, 
	the parallelogram identity yields
	\[
		\|\Th_3 - \sfI\|^2_2 + \frac14\|\Th_1 - \Th_2\|^2_2
		= 
		\frac12\big( \|\Th_1 - \sfI\|_2^2 + \|\Th_2 - \sfI\|_2^2 \big) = M^2.
	\]
	Hence, using that $M = \inf_{\Th\in\sC} \|\Th - \sfI\|_2$ we get
	\[
		\|\Th_1 - \Th_2\|^2_2 
		= 
		4\big( M^2 - \|\Th_3 - \sfI\|^2_2 \big) \le 0.
	\]
	Therefore, we conclude that $\Th_1 = \Th_2$. 	
\end{proof}
Now we have all the tools to formulate and prove the main result of this paper.
\begin{thm}\label{thm:main}
	Let $\sfH$ be an operator as in Hypothesis~\ref{hyp:1}
	and let $\Th_\star\in\sC$ be the unique minimiser
	for the problem~\eqref{eq:problem}. Then the following hold.
	\begin{myenum}
		\item If $\Th_\star\in\sC^\circ$, 	
		then $\Th_\star$ is the minimally anisotropic metric for $\sfH$
		and its characteristic vector $\aa = \aa(\Th_\star)$ satisfies
		the Euler-Lagrange equations
		\begin{equation}\label{eq:EL}
			\sum_{m=1}^N |\langle\phi_m,\phi_n\rangle|^2	\aa_m
			+ 
			\sum_{m\ne n} |\langle\phi_m,\phi_n\rangle|^2 = 0, \qquad \forall\, n\in\dK.
		\end{equation} 

		\item 
		If the system of equations~\eqref{eq:EL}
		has a solution $\aa\in \ell^2(\dK;(-1,\infty))$,
		then such solution is unique in $\ell^2(\dK;(-1,\infty))$ 
		and the operator
		$\Th\in\sC^\circ$ with the characteristic vector $\aa(\Th) = \aa$ coincides with $\Th_\star$.
		
		\item If $\Th_\star\in \p\sC$,
		then the minimally anisotropic metric
		operator for $\sfH$ does not exist.
	\end{myenum}	 
\end{thm}
\begin{proof}
	(i) Clearly, $\Th_\star$ is a metric operator for $\sfH$ by~\eqref{eq:Theta}. 
	By Lemma~\ref{lem:HS} it can be represented as $\Th_\star = \sfI + \sfK$
	with a Hilbert-Schmidt $\sfK$. Moreover, Proposition~\ref{prop:main} implies
	$\|\Th_\star -\sfI\|_2 < \|\Th - \sfI\|_2$ for any $\Th\in\sC^\circ$. Thus, 
	$\Th_\star$ is minimally anisotropic.
	
	Now we derive the Euler-Lagrange equations in~\eqref{eq:EL}.
	The condition of minimality $\|\sfI - \Th_\star\|_2$ implies that the functions
	\begin{equation}\label{eq:fn}
		f_n(\eps) := 
		\big\| \sfI - \Th_\star + \eps\langle\phi_n,\cdot\rangle\phi_n\big\|_2^2,
		\qquad \forall\,n\in\dK,
	\end{equation}
	satisfy $\df_n(0) = 0$. 
	That is we have
	\begin{equation}\label{eq:fn_prime}
	\begin{split}
		\df_n(0) 
		& = 
		\frac{\dd}{\dd\eps}
		\Big(
			\| \sfI - \Th_\star \|_2^2 
			+ 
			2\eps	\Tr\big[( \sfI - \Th_\star ) \phi_n \langle\phi_n,\cdot\rangle \big]
			+	
			\eps^2 \| \phi_n \langle \phi_n, \cdot\rangle \|^2_2
		\Big)
		\Big|_{\eps = 0}\\
		& =
		2\Tr\big[ (\sfI - \Th_\star)
		\phi_n \langle\phi_n,\cdot\rangle \big] 
		= 0, 
		\qquad\forall\, n\in \dK.
	\end{split}	
	\end{equation}
	Hence, using the shorthand notation $\aa_k = \aa_k(\Th_\star)$ we get
	\[
	\begin{split}
		0 & = -\frac{\df_n(0)}{2} 
		= 
		\Tr\left(
		\sum_{k=1}^N 
		\big(\phi_k + \aa_k\phi_k -\psi_k\big)  
		\langle\phi_k,\phi_n\rangle	
		\langle\phi_n,\cdot\rangle
		\right)\\
		& =
		\sum_{k=1}^N
		\Tr
		\Big( 
		\big( \aa_k\phi_k + \phi_k - \psi_k \big)
		\langle \phi_k, \phi_n \rangle
		\langle \phi_n, \cdot \rangle
		\Big)
		=
		\sum_{k=1}^N 
		\big\langle\phi_n,\aa_k\phi_k+ \phi_k-\psi_k\big\rangle
		\langle\phi_k, \phi_n\rangle\\
		& =
		\sum_{k=1}^N  |\langle\phi_k, \phi_n\rangle|^2 + 
		\sum_{k=1}^N \aa_k |\langle\phi_k,\phi_n\rangle|^2 - \|\phi_n\|^2 
		= 
		\sum_{k\ne n} |\langle\phi_k, \phi_n\rangle|^2 + 
		\sum_{k=1}^N \aa_k|\langle\phi_k, \phi_n\rangle|^2. 
	\end{split}
	\]
	\noindent 
	(ii)
	Let $\aa' \in \ell^2(\dK;(-1,\infty))$ be a solution of~\eqref{eq:EL}. Consider the operator $\Th_1\in\sC^\circ$
	with the characteristic vector $\aa(\Th_1) = \aa'$.
	As in the proof of (i), one has
	$\Tr((\sfI - \Th_1)\phi_n\langle\phi_n,\cdot\rangle) = 0$
	for all $n\in\dK$.
	Let $\Th_2\in\sC^\circ$ be such that $\Th_2\ne \Th_1$.
	Hence, we obtain
	\[
		\|\sfI - \Th_2\|^2_2 
		= 
		\|\sfI - \Th_1 + \Th_1 - \Th_2\|^2_2
		=
		\|\sfI - \Th_1\|^2_2 + \|\Th_1 - \Th_2\|^2_2
		> 
		\|\sfI - \Th_1\|^2_2. 
	\]
	Thus, $\Th_1$ is indeed the minimally anisotropic metric.
	Existence of another solution $\aa''\in\ell^2(\dK;(-1,\infty))$ for~\eqref{eq:EL} contradicts uniqueness of the minimally
	anisotropic metric.

	\noindent (iii)
	By~\eqref{eq:Theta} we infer that $\Th_\star$ is  \emph{not} a metric 
	for $\sfH$. Proposition~\ref{prop:main} yields that
	any metric $\Th\in\sC^\circ$ for $\sfH$ 
	satisfies the inequality $M := \|\Th_\star - \sfI\|_2 < \|\Th - \sfI\|_2$.
	If $N < \infty$, then it is easy to construct a sequence of metrics $\Th_m\in\sC^\circ$, $m\in\dN$,
	which converges in the Hilbert-Schmidt norm to $\Th_\star$, and we omit this construction.
	If $N = \infty$, then we consider the following sequence of operators
	\[
		\Th_m 
		:= 
		\sum_{ n = 1 }^m \big[ 1 + \aa_n(\Th_\star) + e^{-m} \big] \phi_n\langle\phi_n, \cdot
		\rangle
		+
		\sum_{ n = m+1 }^\infty \phi_n\langle\phi_n, \cdot\rangle.
	\]
	It is easy to check that $\Th_m\in\sC^\circ$. Moreover, we get
	\[
	\begin{split}
		M & \le \limm \|\Th_m - \sfI\|_2
		\le \|\Th_\star - \sfI\|_2 + \limm\|\Th_m - \Th_\star\|_2\\
		&  
		= 
		M + 
		\|\sfA^{-1}\|^2
		\limm	\|\sfA^*(\Th_m - \Th_\star)\sfA\|_2 \\
		&= 
		M + 
		\|\sfA^{-1}\|^2\limm\Big (m e^{-2m} + \textstyle{\sum}_{n=m+1}^\infty |\aa_n(\Th_\star)|^2\Big)^\frac12
		= M.
	\end{split}	
	\]
	Hence, we conclude that $\limm \|\Th_m - \sfI\|_2 = \|\Th_\star - \sfI\|_2$ and thus
	the minimally anisotropic metric does not exist.
\end{proof}
Finally, we provide a sufficient condition for the
unique minimiser $\Th_\star\in \sC$ of~\eqref{eq:problem} to be indeed a metric for $\sfH$,
that is for $\Th_\star \in \sC^\circ$ to hold.
\begin{prop}\label{prop:existence}
	Let $\sfH$ be an operator as in Hypothesis~\ref{hyp:1}. In addition, assume that 
	\begin{equation}\label{eq:condition}
		\sum_{n=1}^N\sum_{m\ne n}
		|\langle\phi_n,\phi_m\rangle|^2 < 1\,.
	\end{equation} 
	Then
	the unique minimiser $\Th_\star$ of~\eqref{eq:problem} satisfies $\Th_\star\in \sC^\circ$, thus, being a metric for $\sfH$.
\end{prop}	  	
\begin{proof}
	Clearly, the decomposition $\sC = \p\sC\cup\sC^\circ$
	with $\p\sC\cap\sC^\circ = \varnothing$
	holds. For any $\Th\in\p\sC$ there exists $n_0\in\dK$ such that 
	$\aa_{n_0}(\Th) = -1$ and we get
	\[
		\|\Th - \sfI\|_2 
		\ge
		\|\Th - \sfI\| 
		\ge 
		\frac{\big|\big\langle
			\psi_{n_0}, (\Th -\sfI)\psi_{n_0}
			\big\rangle\big|}{\|\psi_{n_0}\|^2}	
		=
		\frac{\|\psi_{n_0}\|^2}{\|\psi_{n_0}\|^2}	
		= 1\,.
	\]
	For the constant-coefficient metric operator
	$\Th_0 = \sum_{n=1}^N \phi_n\langle\phi_n,
		\cdot\rangle \in \sC^\circ$
	we compute
	\[
		\left(\sfI - \Th_0\right)^2
		=
		\sum_{n=1}^N\sum_{m=1}^N 
		(\psi_n - \phi_n)
		\langle\phi_n,\psi_m - \phi_m\rangle
		\langle\phi_m,\cdot\rangle\,.
	\]
	Hence, we obtain
	\[
	\begin{split}
		\|\sfI - \Th_0\|_2^2
		& =
		\Tr
		\left(\sfI - \Th_0\right)^2
		=
		\sum_{n=1}^N\sum_{m=1}^N 
		\langle\phi_n,\psi_m - \phi_m\rangle
		\langle\phi_m,\psi_n - \phi_n\rangle\\
		& \le
		\sum_{n=1}^N\sum_{m=1}^N 
		|\dl_{nm} - \langle\phi_n,\phi_m\rangle|
		|\dl_{nm} - \langle\phi_m,\phi_n\rangle|\\
		& =	
		\sum_{n=1}^N\sum_{m=1}^N 
		|\dl_{nm} - \langle\phi_n,\phi_m\rangle|^2
		 =
		\sum_{n=1}^N\sum_{m\ne n}
		|\langle\phi_n,\phi_m\rangle|^2 < 1\,.\\
	\end{split}	
	\]
	Therefore, we have
	$\|\sfI - \Th_\star\|_2 \le	\|\sfI - \Th_0\|_2 <1$
	and infer that $\Th_\star \notin\p\sC$.
	Thus, we conclude $\Th_\star \in \sC^\circ$.
\end{proof}	

\begin{remark}
	The sum in~\eqref{eq:condition} can be 
	interpreted as squared Hilbert-Schmidt norm
	of the linear operator in the Hilbert space $\ell^2(\dK)$
	induced by the matrix $\{\dl_{nm} - \langle\phi_n,\phi_m\rangle\}_{n,m\in\dK}$.
	We point out that the same matrix appears in~\cite[\S VI.3]{Gohberg-Krein_1969}.
	In particular, for $\dK = \dN$, finiteness of its Hilbert-Schmidt
	is necessary and sufficient for $\Phi = \{\phi_n\}_{n=1}^{\infty}$
	to be a Bari basis, provided that $\Phi$
	is $\omg$-linearly independent.
\end{remark}

\section{Finite-dimensional examples}	
In this section we provide a couple of toy examples in two- and four-dimensional Hilbert spaces.
The example in $\dC^4$ is very special and its aim is to show that the minimally anisotropic metric
indeed does not always exist.

\subsection{$2\tm 2$ example}
Let the vectors $\phi_1,\phi_2\in\dC^2$ be linearly independent and normalised 
as $\|\phi_1\|_{\dC^2} =  \|\phi_2\|_{\dC^2} =  1$.
We select the vectors $\psi_1, \psi_2 \in \dC^2$ so that 
$\langle \psi_n, \phi_m \rangle_{\dC^2} = \dl_{nm}$, $n,m\in\{1,2\}$.
Any matrix
\[
	\sfH = 
	\lm_1 \psi_1 \langle \phi_1, \cdot \rangle	+ \lm_2\psi_2 \langle \phi_2, \cdot \rangle,
	\qquad -\infty < \lm_1< \lm_2 +\infty,
\]
can be viewed as a quasi-self-adjoint operator in the Hilbert space $\dC^2$
satisfying Hypothesis~\ref{hyp:1}. 
According to~\eqref{eq:Theta}, any metric 
for $\sfH$ can be decomposed as
\[
	\Th = 
	\big[1+\aa_1\big]\phi_1\langle\phi_1,\cdot\rangle + 
			\big[1+\aa_2\big]\phi_2\langle\phi_2,\cdot\rangle,
\]
with $\aa_1,\aa_2\in (-1,\infty)$.
Using the shorthand notation $\gg := |\langle\phi_1,\phi_2\rangle_{\dC^2}|^2$
we can write the system of Euler-Lagrange equations~\eqref{eq:EL}
as follows
\[
	\begin{cases}
		\gg\aa_1 + \aa_2    = - \gg,\\
		\aa_1    + \gg\aa_2 = - \gg.
	\end{cases}	
\]
Solving the above linear system, we find
\begin{equation}\label{eq:aa12}
	\aa_1 = \aa_2 = -\frac{\gg}{1+ \gg} > -1.
\end{equation}
Hence, by Theorem~\ref{thm:main}\,(ii) the minimally anisotropic
metric always exists and is given by
\[
	\Th_\star = \frac{\phi_1\langle\phi_1,\cdot\rangle}{1+\gg}
	+
	\frac{\phi_2\langle\phi_2,\cdot\rangle}{1+\gg}.
\] 
In the special case
\[
	\phi_1 = \begin{pmatrix} 1 \\ 0\end{pmatrix},
	\quad
	\phi_2 = \frac{1}{\sqrt{2}}
	\begin{pmatrix} 1\\ 1\end{pmatrix}
	\and
	\lm_1 = 1,\quad \lm_2 = 2,
\]
we have
\[
	\psi_1 = \begin{pmatrix} 1 \\ -1\end{pmatrix}
	\and
	\psi_2 =	\begin{pmatrix} 0 \\ \sqrt{2}\end{pmatrix}.
\]
Thus, the quasi-self-adjoint Hamiltonian $\sfH$ is given by
\[
	\sfH = 
	\psi_1\langle\phi_1,\cdot\rangle
	+
	2\psi_2\langle\phi_2,\cdot\rangle
	=
	\begin{pmatrix}
	1 &  0\\
	-1 & 0
	\end{pmatrix}
	+
	\begin{pmatrix}
	0 & 0\\
	2 & 2
	\end{pmatrix}
	= 
	\begin{pmatrix}
	1 & 0\\
	1 & 2
	\end{pmatrix}.
\]
The definition of $\gg$ and 
the formula~\eqref{eq:aa12} yield $\gg = \frac12$, $\aa_1 = \aa_2 = -\frac13$.
Thus, the minimally anisotropic metric is explicitly given by
\[
	\Th_\star =
	\frac23\phi_1\langle\phi_1,\cdot\rangle
	+
	\frac23\phi_2\langle\phi_2,\cdot\rangle
	= 
	\frac{2}{3}
	\begin{pmatrix}
	1 &  0\\
	0 & 0
	\end{pmatrix}
	+
	\frac{1}{3}
	\begin{pmatrix}
		1 & 1\\
		1 & 1
	\end{pmatrix}
	=
	\begin{pmatrix}
	1 &  \frac13\\
	\frac13 & \frac13
	\end{pmatrix}.
\]
\subsection{$4\tm 4$ example}
Let us fix $x\in (0,1)$, set $y := \sqrt{1-x^2}$, and consider the following
normalised vectors in $\dC^4$
\[
	\phi_1 = (1,0,0,0)^\top,\qquad
	\phi_2 = \big(y,x,0,0\big)^\top,\qquad
	\phi_3 = \big(y,0,x,0\big)^\top,\qquad
	\phi_4 = \big(y,0,0,x\big)^\top.
\]
We will show that for all sufficiently small
$x\in (0,1)$ the minimally anisotropic
metric for any quasi-self-adjoint Hamiltonian
$\sfH = \sum_{j=1}^4\lm_j\phi_j\langle\phi_j,\cdot\rangle$
with $\lm_j\in\dR$ ($\lm_i\ne \lm_j$ for $i\ne j$)
does not exist.
Suppose that the minimally anisotropic
metric $\Th_\star\in\sC^\circ$ with the characteristic
vector $\aa = (\aa_1,\aa_2,\aa_3,\aa_4)^\top$ exists. Taking the symmetries
into account, we conclude that $\aa_2 = \aa_3 = \aa_4$.
By Theorem~\ref{thm:main} we can write the system of Euler-Lagrange equations~\eqref{eq:EL}
with the notation $a := \aa_1$ and $b := \aa_2 = \aa_3 = \aa_4$ as follows 
\[
\begin{cases}
	a          + 3y^2b                 = - 3y^2,\\
	y^2a   + \big(1 + 2y^4\big)b = - y^2 - 2y^4.
\end{cases}	
\]
The system can be simplified as
\[
\begin{cases}
	a  + 3y^2b                    =  - 3y^2,\\
	a  + \frac{1+2y^4}{y^2}b = - 1 - 2y^2.
\end{cases}	
\]
Hence, we derive an equation on $b$ 
\[
	\frac{y^4-1}{y^2}b = 1-y^2.
\]
Finally, we get
\[
	a = -\frac{3y^2}{y^2+1}
	\and
	b = -\frac{y^2}{y^2+1}
\]
For a sufficiently small $x\in (0,1)$,
the difference
$1-y$ is arbitrarily small, we have $a < -1$ and therefore the minimally
anisotropic metric does not exist.

\section{An $\infty$-dimensional example: $\PT$-symmetric Robin Laplacian}
The feasibility of the construction of the minimally anisotropic metric
for quantum systems with $N=\infty$ may be illustrated via the $\PT$-symmetric Robin Laplacian on an interval. For this purpose, let us set $\cJ := (-\frac{\pi}{2},\frac{\pi}{2})$ 
and define the \emph{conjugation}
and \emph{parity} operators on the Hilbert space 
$(L^2(\cJ), \langle\cdot,\cdot\rangle_\cJ)$ as
\begin{equation}\label{eq:PT}
	(\cT\psi)(x) := \ov{\psi(x)}
	\and
	(\cP\psi)(x) := \psi(-x).
\end{equation}
Introduce the m-sectorial operator $\sfH_\bb$, $\bb\in\dR$, on $L^2(\cJ)$ as
\[
	\sfH_\bb\psi := -\psi'', \qquad
	\dom\sfH_\bb := 
	\big\{ \psi \in H^2(\cJ) 
		\colon \psi'\left(\pm \tfrac{\pi}{2}\right) + \ii\bb
		\psi\left( \pm \tfrac{\pi}{2}\right) = 0
	\big\}.
\label{eq:op}
\]
The spectral theory of $\sfH_\bb$ is developed in~\cite{KBZ, K4, KSZ}. 
Below we recall basic spectral properties of this operator.
\begin{prop}[{\cite{KBZ},\cite[Prop.~2.4]{KSZ}}]\label{prop:Hb_basics}
	Let the operator $\sfH_\bb$, $\bb\in\dR$, be as in~\eqref{eq:op}. 
	Let $\cP$ and $\cT$ be as in~\eqref{eq:PT}. Then the following hold.
	\begin{myenum}
		\item 
		$\sfH_\bb^* = \sfH_{-\bb}, 
		\quad 
		\sfH_\bb = \cP \sfH_\bb^* \cP, \quad 
		\PT \sfH_\bb\subseteq \sfH_\bb \PT$. 
		In particular, $\sfH_\bb$ is $\cP$-self-adjoint.
		\item
		$\s(\sfH_\bb) = \cup_{n \in \dN_0}\{ \lm_n \} \subset \R$, 	
		where $\lm_0 = \bb^2$ and $\lm_n = n^2$, $n\in\dN$.
		\item If $\bb \notin \dZ\sm\{0\}$, 
		then $\sfH_\bb$ satisfies Hypothesis~\ref{hyp:1}.
		The eigenfunctions $\{\psi_n^\bb\}_{n=0}^\infty$
		and $\{\phi_n^\bb\}_{n=0}^\infty$
		of $\sfH_\bb$ and $\sfH_\bb^*$, respectively, corresponding to 
		$\{\lm_n\}_{n \in \dN_0}$ read as
		\begin{align*}\label{efunc}
			\psi_0^\bb(x) 
			& = 
			A_0 e^{-\ii \bb (x+a)},\!
			&\psi_n^\bb(x)& = A_n\left[ \cos (n(x+a)) -	\frac{\ii \bb}{n}\sin(n(x+a))\right],\!\!\!
			 &n& \in \dN\,,\\
			\phi_0^\bb(x) 	&= B_0 e^{\ii \bb (x+a)},\!
				&\phi_n^\bb(x) & = 
			B_n\left[\cos (n (x+a)) +
			\frac{\ii \bb}{n}\sin(n(x+a))\right],\!\!\! 
			&n& \in \dN\,,
		\end{align*}
		where $a = \frac{\pi}{2}$ and the real positive constants 
		$\{A_n\}_{n=0}^\infty$
		and $\{B_n\}_{n=0}^\infty $ are chosen so that
		$(\psi_n^\beta,\phi_n^\beta) = \dl_{nm}$
		and $\|\phi_n^\beta\| = 1$. In particular,
		\[
			B_0 = \frac{1}{\sqrt{\pi}}
			\and
			B_n = 
			\sqrt{\frac{2}{\pi}}
			\frac{n}{\sqrt{n^2+ \beta^2}},
			\quad n\in\dN\,.
		\]
		\item If $ \bb \in \dZ\sm\{0\}$, then 
		the eigenvalue
		$\lm_0$ of $\sfH_\bb$ has the geometric multiplicity one and the algebraic multiplicity two, 
		and all the other eigenvalues are simple.
	\end{myenum}	
\end{prop}
Proposition~\ref{prop:Hb_basics}\,(iii), representation~\eqref{eq:Theta}, and Lemma~\ref{lem:HS} imply that for $\bb\notin\dZ\sm\{0\}$ all the metrics for $\sfH_\bb$, that admit the representation
$\sfI + \sfK$ with a Hilbert-Schmidt $\sfK$, can be written as
\[
	\Th = \sum_{n=0}^\infty 
	\big[1+\aa_n\big]
	\phi_n^\bb\langle\phi_n^\bb,\cdot\rangle_{\cJ},
	\qquad \aa\in\ell^2(\dN_0;\cI),~~\min\aa > -1.
\]
The coefficients $a_{nm}:=  \langle\phi_n^\bb,\phi_m^\bb\rangle$
can be explicitly computed.
Clearly, $a_{nn} = 1$ holds for all $n\in\dN_0$
and for $n,m \ge 1$, $n \ne m$ we have
%
%
\[
\begin{split}
	a_{nm}
	& = 
	\frac{B_nB_m}{nm}
	\int_0^\pi
	\big[n\cos (n x) - 
	\ii \bb \sin(nx)\big]
	\big[m\cos (m x) +
	\ii \bb \sin(mx)\big]
	 \dd x \\
	& =
	\frac{B_nB_m \ii \bb}{2nm}
	\int_0^\pi
	\big[(n-m)\sin ((n + m) x) - (n+m)\sin((n-m) x)\big]\dd x \\
	& =
	\frac{B_nB_m(1-e^{\ii\pi (n+m)}) \ii \bb}{2nm}
	\left(
	\frac{n-m}{n+m}
	-
	\frac{n+m}{n-m}
	\right)\\
	& =
	\frac{B_nB_m(1- e^{\ii\pi (n+m)}) \ii \bb}{2nm}
	\frac{-4mn}{n^2-m^2}
	=
	\frac{2B_nB_m\ii \bb(1-e^{\ii\pi (n+m)}) }{m^2-n^2}\,.
\end{split}	
\]
For $n \ge 1$ we get
\[
\begin{split}
	a_{n0}
	& = 
	\frac{B_nB_0}{n}
	\int_0^\pi
	\big[n\cos (n x) - 
		\ii \bb \sin(nx)\big] e^{\ii \bb x}
	\dd x \\
	& =
	\frac{B_nB_0}{n}
	\left(
	\frac{2(-1)^n n\bb\sin(\pi \bb)}{\bb^2-n^2}
	+ 
	\frac{2\ii(-1)^n n\bb [(-1)^n  - \cos(\pi\bb)]}{\bb^2-n^2}
	\right)
	\\
	& =
	\frac{2B_nB_0\ii\bb\left(
		1  - e^{\ii\pi(\bb +n)}\right)}{\bb^2-n^2}\,.\\
\end{split}	
\]
The remaining coefficients can be recovered
via the relation $a_{nm} = \ov{a}_{mn}$.
The Euler-Lagrange system in~\eqref{eq:EL} reduces to
\begin{equation}\label{eq:EL_Rob}
	\sum_{m=0}^\infty |a_{nm}|^2\aa_m(\Th_\star)
	+\sum_{m\ne n} |a_{nm}|^2 = 0,\qquad n\in\dN_0.
\end{equation}
It seems however that this system can not be explicitly solved. 

On the other hand, using 
the above expressions for $a_{nm}$
and the formulae for $B_n$ in Proposition~\ref{prop:Hb_basics}\,(iii)
we obtain the following bounds on the coefficients in~\eqref{eq:EL_Rob}:
\[
	|a_{nm}|^2 \le \frac{64\bb^2}{\pi^2(m^2-n^2)^2},
	\qquad
	|a_{n0}|^2 \le 
	\frac{64\bb^2}{\pi^2(\bb^2-n^2)^2},
	\qquad n,m\in\dN,\, n\ne m	\,.
\]
Using these bounds we get for any 
$\bb \in (0,\frac12)$
\[
\begin{split}
	\sum_{n=0}^\infty\sum_{m\ne n} |a_{nm}|^2 & \le 
	\frac{64\bb^2}{\pi^2}
	\left(
	\sum_{n=1}^\infty\frac{2}{(\bb^2-n^2)^2}
	+ \sum_{n=1}^\infty
	\sum_{m > n}
	\frac{2}{(m^2-n^2)^2}\right)\\
	& <
	\frac{64\bb^2}{\pi^2}
	\left(
	\sum_{n=1}^\infty\frac{2}{(n^2-1/4)^2}
	+ \sum_{d = 1}^\infty
	\sum_{n=1}^\infty
	\frac{2}{((n+d)^2-n^2)^2}\right)\\
	& \le 
	\frac{64\bb^2}{\pi^2}
	\left(2\pi^2 - 16 +
	\frac12
	\sum_{d = 1}^\infty
	\frac{1}{d^2}
	\sum_{n=1}^\infty\frac{1}{n^2}\right)
	 =
	\frac{64\bb^2}{\pi^2}\left(2\pi^2 - 16 + \frac{\pi^4}{72}\right).
\end{split}	
\]
Hence, for all $\bb\in(0,\frac12)$ small enough we have $\sum_{n=0}^\infty\sum_{m\ne n} |a_{nm}|^2 < 1$	
and Proposition~\ref{prop:existence} yields that the unique minimizer of~\eqref{eq:problem} in this
particular
setting is indeed a metric operator for $\sfH_\bb$.

Another possibility to select a metric operator relies on the concept of the 
$\cC$-symmetry.
\begin{dfn}[Charge operator $\cC$] 
	Assume that $\sfH$ in $L^2(\cJ)$ is $\cP$-self-adjoint. We say that $\sfH$ possesses the property of $\cC$-symmetry, if there exists a bounded linear
	operator $\cC$ in $L^2(\cJ)$ such that 
	$[\sfH,\cC] = 0$, $\cC^2 = \sfI$, and $\cP\cC$ is a metric  for $\sfH$.
\end{dfn}
By~\cite[Sec.~4.3]{KSZ} the operator $\sfH_\bb$
with $\bb \notin \dZ\sm\{0\}$ possesses a
$\cC$-symmetry with the uniquely determined
charge
operator $\cC\colon L^2(\cJ)\arr L^2(\cJ)$ through the identities~$\cC^2 = \sfI$
and $\cC = \cP\Th$ with a metric operator $\Th$
as in~\eqref{eq:Theta}.
The metric for $\sfH_\beta$ induced by the charge operator $\cC$
is given by
$\Th_\cC := \cP\cC = \sfI + \sfK_\cC$,
where $\sfK_\cC$ is an integral
operator in $L^2(\cJ)$ with the kernel
\[
	\sfk_\cC(x,y) = \bb e^{-\ii \bb (y-x)}
	\left [\tan\big(\tfrac{\pi\bb}{2}\big)-
	\ii\,\sign(y-x)\right].
\]	
It is easy to see that $\|\sfK_\cC\|_2 < \infty$. 
Consider the function
\[
	\dR\ni\eps\arr f_0(\eps) = \big\| 
		\sfK_\cC 
		+ \eps\phi_0^\bb\langle\phi_0^\bb,\cdot\rangle
	\big\|_2^2.
\]
To show that $\Th_\cC$
is not the minimiser for~\eqref{eq:problem} it suffices to check that $\df_0(0) \ne 0$.
Differentiating $f_0(\cdot)$ at the point $\eps = 0$ we find
\[
\begin{split}
	\df_0(0) & = 2\Tr\big(\sfK_\cC\phi_0^\bb
	\langle\phi_0^\bb,
	\cdot\rangle\big) = 2\big\langle
	\phi_0^\bb,\sfK_\cC\phi_0^\bb\big\rangle_{\cJ}
	= 2\int_\cJ\int_\cJ
	\sfk_\cC(x,y)\ov{\phi_0^\bb(x)}\phi_0^\bb(y)\dd x \dd y\\
	& =
	\frac{2\bb}{\pi}
	\int_\cJ\int_\cJ
	\left [\tan\big(\tfrac{\pi\bb}{2}\big)-
	\ii\,\sign(y-x)\right] \dd x \dd y
	=
	2\bb\pi\tan\big(\tfrac{\pi\bb}{2}\big)\ne 0.
\end{split}
\]

\section*{Acknowledgements}	
D.\,K. was partially supported by the GA\v{C}R grant No.~18-08835S and by FCT (Portugal) through project PTDC/MAT-CAL/4334/2014. 	
V.\,L. acknowledges the support by the GA\v{C}R  grant No.~17-01706S.
M.\,Z. acknowledges the support by the GA\v{C}R grant No.~16-22945S.

%
\bibliographystyle{amsplain}

\providecommand{\bysame}{\leavevmode\hbox to3em{\hrulefill}\thinspace}
\providecommand{\MR}{\relax\ifhmode\unskip\space\fi MR }
\providecommand{\MRhref}[2]{%
  \href{http://www.ams.org/mathscinet-getitem?mr=#1}{#2}
}
\providecommand{\href}[2]{#2}

\end{document}